\documentclass[sigconf]{acmart}
\usepackage{booktabs} 
\usepackage{amsmath}
\usepackage[linesnumbered,ruled]{algorithm2e}
\usepackage{multicol,tabularx,capt-of}
\usepackage{multirow}
\usepackage{flushend}
\setcopyright{rightsretained}


\copyrightyear{2017} 
\acmYear{2017} 
\setcopyright{acmcopyright}
\acmConference[Compute '17]{10th Annual ACM India Compute Conference}{November 16--18, 2017}{Bhopal, India}
\acmPrice{15.00}
\acmDOI{10.1145/3140107.3140114}
\acmISBN{978-1-4503-5323-6/17/11}

\begin{document}
\title{A Novel Weighted Distance Measure for Multi-Attributed Graph}
\author{Muhammad Abulaish, \textit{SMIEEE}}
\affiliation{%
	\department{Department of Computer Science}
	\institution{South Asian University}
	\city{New Delhi-21}
	\country{India}}
\email{abulaish@sau.ac.in}
\author{Jahiruddin}
\authornote{\textit{To whom correspondence should be made}}
\affiliation{%
	\department{Department of Computer Science}
	\institution{Jamia Millia Islamia}
	\city{New Delhi-25}
	\country{India}}
\email{jahir.jmi@gmail.com}

\renewcommand{\shortauthors}{Abulaish et al.}
\begin{abstract}
Due to exponential growth of complex data, graph structure has become increasingly important to model various entities and their interactions, with many interesting applications including, bioinformatics, social network analysis, etc. Depending on the complexity of the data, the underlying graph model can be a simple directed/undirected and/or weighted/un-weighted graph to a complex graph (aka multi-attributed graph) where vertices and edges are labelled with multi-dimensional vectors. In this paper, we present a novel weighted distance measure based on weighted Euclidean norm which is defined as a function of both vertex and edge attributes, and it can be used for various graph analysis tasks including classification and cluster analysis. The proposed distance measure has flexibility to increase/decrease the weightage of edge labels while calculating the distance between vertex-pairs. We have also proposed a \emph{MAGDist} algorithm, which reads multi-attributed graph stored in CSV files containing the list of vertex vectors and edge vectors, and calculates the distance between each vertex-pair using the proposed weighted distance measure. Finally, we have proposed a multi-attributed similarity graph generation algorithm, \emph{MAGSim}, which reads the output of \emph{MAGDist} algorithm and generates a similarity graph that can be analysed using classification and clustering algorithms. The significance and accuracy of the proposed distance measure and algorithms is evaluated on Iris and Twitter data sets, and it is found that the similarity graph generated by our proposed method yields better clustering results than the existing similarity graph generation methods.   
\end{abstract}
%
%
%
\maketitle
\section{Introduction}
Due to increasing popularity of Internet and Web2.0, the User-Generated Contents (UGCs) are growing exponentially. Since most of the UGCs are not independent, rather linked, the graph data structure is being considered as a suitable mathematical tool to model the inherent relationships among data. Simple linked data are generally modelled using simple graph $G=(V, E)$, where  $V$ is the set of vertices representing key concepts or entities and $E \subseteq  V \times V$ is the set of links between the vertices representing the relationships between the concepts or entities. Depending on the nature of data to be modelled, the graph $G$ could be weighted/un-weighted or directed/undirected, and it may have self-loops.  However, there are many complex data like online social networks, where an entity is characterized by a set of features and multiple relationships exist between an entity-pair. To model such data, the concept of multi-attributed graph can be used wherein each vertex is represented by an $n$-dimensional vector and there may be multiple weighted edges between each pair of vertices. In other words, in a multi-attributed graph, both vertices and edges are assigned multiple labels.   
One of the important tasks related to graph data analysis is to decompose a given graph into multiple cohesive sub-graphs, called clusters, based on some common properties. The clustering is an unsupervised learning process to identify the underlying structure (in the form of clusters) of data, which is generally based on some similarity/distance measures between data elements.
Graph clustering is special case of clustering process which divides an input graph into a number of connected components (sub-graphs) such that intra-component edges are maximum and inter-components edges are minimum. Each connected component is called a cluster (aka community) \cite{Girvan02}. Though graph clustering has received attention of many researchers and a number of methods for graph clustering has been proposed by various researchers \cite{Schaeffer2007}, to the best of our knowledge, the field of clustering multi-attributed graph is still unexplored. 
 
Since similarity/distance measure is the key requirement for any clustering algorithm, in this paper, we have proposed a novel weighted distance measure based on weighted Euclidean norm to calculate the distance between the vertices of a multi-attributed graph. The proposed distance measure considers both vertex and edge weight-vectors, and it is flexible enough to assign different weightage to different edges and scale the overall edge-weight while computing the weighted distance between a vertex-pair. We have also proposed a \emph{MAGDist} algorithm that reads the lists of vertex and edge vectors as two separate CSV files and calculates the distance between each vertex-pairs using the proposed weighted distance measure. Finally, we have proposed a multi-attributed similarity graph generation algorithm, \emph{MAGSim}, which reads the output of \emph{MAGDist} algorithm and generates a similarity graph. In other words, \emph{MAGDist} and \emph{MAGSim} algorithms can be used to transform a multi-attributed graph into a simple weighted similarity graph, wherein a single weighted edge exists between each vertex-pair. Thereafter, the weighted similarity graph can be analysed using existing classification and clustering algorithms for varied purposes. The efficacy of the proposed distance measure and algorithms is tested over the well-known Iris data set and a Twitter data set related to three different events. The proposed similarity graph generation approach is compared with other existing similarity graph generation methods like Gaussian kernel and $k$-nearest neighbours methods, and it is found that our proposed approach yields better clustering results in comparison to the existing methods. Moreover, the proposed distance measure can be applied to any graph data where both vertices are edges are multi-dimensional real vectors. In case, the data is not linked, i.e., edges do not exist between the edges, the proposed distance measure simple works as an Euclidean distance between the node vectors.
 
The rest of the paper is organized as follows. Section 2 presents a brief review on various distance measures and graph clustering methods. Section 3 presents the formulation of our proposed weighted distance measure for multi-attributed graph. It also presents the founding mathematical concepts and formal descriptions of the \emph{MAGDist} and \emph{MAGSim} algorithms. Section 4 presents the experimental setup and evaluation results. Finally, section 5 concludes the paper with future directions of work.
\section{Related Work}
Multi-attributed graph is used to model many complex problems, mainly those in which objects or entities are characterized using a set of features and linked together in different ways. For example, an online social network user can be characterized using a set of features like ID, display name, demographic information, interests, etc. Similarly, two person in an online social network may be associated through different relationships like friendship, kinship, common interests, common friends, etc \cite{Sajid2013}. In \cite{kolar14a}, the authors proposed a new principled framework for estimating graphs from multi-attribute data. Their method estimates the partial canonical correlations that naturally accommodate complex vertex features instead of estimating the partial correlations. However, when the vertices of a multi-attributed graph have different dimensions, it is unclear how to combine the estimated graphs to obtain a single Markov graph reflecting the structure of the underlying complex multi-attributed data. In \cite{Katenka2011},  Katenka and Kolaczyk proposed a method for estimating association networks from multi-attribute data using canonical correlation as a dependence measure between two groups of attributes.

Clustering is an unsupervised learning technique which aims to partition a data set into smaller groups in which elements of a group are similar and that elements from different groups are dissimilar. As a result, clustering has broad applications, including the analysis of business data, spatial data, biological data, social network data, and time series data \cite{Zhou2009}, and a large number of researchers have targeted clustering problem \cite{Agrawal1998, Gibson1998, Ng1994}. Since graph is a popular data structure to model structural as well as contextual relationships of data objects, recently graph data clustering is considered as one of the interesting and challenging research problems \cite{NewGir04, Xu2007}, and it aims to decompose a large graph into different densely connected components. Some of the applications of the graph clustering is community detection in online social networks  \cite{Bhat2014, Bhat2015, Bhat2017}, social bot detection \cite{Fazil2017, Fazil2017a}, spammer detection \cite{Abulaish2015, Sajid2014, Faraz2013, SajidBhat2014}, functional relation identification in large protein-protein interaction networks, and so on. Generally graph clustering methods are based on
the concept of normalized cut, structural density, or modularity\cite{NewGir04, Xu2007}. However, like \cite{Tian2008}, graphs can be partitioned on the basis of attribute similarity in such a way that  vertices having similar attribute vectors are grouped together to form a cluster.  
Mainly graph clustering and simple data clustering differs in the way associations between objects are calculated. In graph clustering, the degree of association between a pair of objects is calculated as closeness between the respective nodes of the graph, generally in terms of number of direct links or paths between them. Whereas, in simple data clustering, the degree of association between a pair of objects is calculated in terms of similarity/distance measure between the vector representations of the objects. 
Both topological structure and vertex properties of a graph play an important role in many real applications. For example, in a social graph, vertex properties can be used to characterize network users, whereas edges can be used to model different types of relationships between the users. 

It can be seen from the discussions mentioned above that most of the graph analysis approaches have considered only one aspect of the graph structure and ignored the other aspects, due to which the generated clusters either have loose intra-cluster structures or reflect a random vertex properties distribution within clusters. However, as stated in \cite{Zhou2009}, ``an ideal graph clustering should generate clusters which have a cohesive intra-cluster structure with homogeneous vertex properties, by balancing the structural and attribute similarities". Therefore, considering both vertex attributes and links for calculating similarity/distance between the pairs of vertices in a graph is one of basic requirements for clustering complex multi-attributed graphs. In this paper, we have proposed a weighted distance function that can transform a multi-attributed graph into a simple similarity graph on which existing graph clustering algorithms can be applied to identify densely connected components. 
\section{Preliminaries and Distance Measures}
In this section, we present the mathematical basis and formulation of the proposed weighted distance measures for multi-attributed graph. Starting with the basic mathematical concepts in subsection \ref{SC1}, the formulation of the proposed distance function along with an example is presented in the subsequent subsections.
\subsection{Founding Mathematical Concepts}\label{SC1}
Since the Linear Algebra concepts \textit{inner product} and \textit{norm} generally form the basis for any distance function, we present a brief overview of these mathematical concepts in this section.
\subsubsection*{Inner Product:}
An \textit{inner product} is a generalized form of the \textit{vector dot product} to multiply vectors together in such a way that the end result is a scalar quantity. If $\vec{a}=(a_1, a_2, ...,a_n)^T$ and $\vec{b}=(b_1, b_2, ...,b_n)^T$ are two vectors in the vector space $\Re^n$, then the \textit{inner product} of $\vec{a}$ and $\vec{b}$ is denoted by $\langle \vec{a}\,, \vec{b}\rangle$ and defined using equation \ref{eqn1} \cite{Olver08}.
\begin{equation}\label{eqn1}
  \langle \vec{a}\,, \vec{b}\rangle = \vec{a} \cdot \vec{b} = a_1b_1 + a_2b_2 + \dots + a_nb_n = \sum_{i=1}^n a_ib_i
\end{equation}
\begin{equation}\label{eqn2}
\begin{split}
  \langle \vec{a}\,, \vec{b}\rangle = a^Tb =
   \begin{bmatrix}
  a_1 & a_2 & \dots & a_n
  \end{bmatrix}
  \begin{bmatrix}
    b_1 \\
    b_2 \\
    \vdots \\
    b_n
  \end{bmatrix} \\
  = \begin{bmatrix}
    a_1b_1 + a_2b_2 + \dots + a_nb_n
  \end{bmatrix}
\end{split}
\end{equation}
Alternatively, \textit{inner product} of $\vec{a}$ and $\vec{b}$ can be defined as a matrix product of row vector $\vec{a}^T$ and column vector $\vec{b}$ using equation \ref{eqn2}. However, as stated in \cite{Olver08}, an inner product must satisfy the following four basic properties, where  $\vec{a}$, $\vec{b}$, and $\vec{c}$ belong to $\Re^n$, and $n$ is a \textit{scalar} quantity.
\begin{enumerate}
  \item  $\langle \vec{a}+\vec{b}\,, \vec{c}\rangle = \langle \vec{a}\,, \vec{c}\rangle + \langle \vec{b}\,, \vec{c}\rangle$
  \item $\langle n\vec{a}\,, \vec{b}\rangle = n\langle \vec{a}\,, \vec{b}\rangle$
  \item $\langle \vec{a}\,, \vec{b}\rangle = \langle \vec{b}\,, \vec{a}\rangle$
  \item $\langle \vec{a}\,, \vec{a}\rangle > 0$  and $\langle \vec{a} \,, \vec{a} \rangle = 0$ iff $\vec{a}=0$
\end{enumerate}
\subsubsection*{Weighted Inner Product:}
The \textit{weighted inner product} is generally used to emphasize certain features (dimensions) through assigning weight to each dimension of the vector space $\Re^n$. If $d_1, d_2,  \dots, d_n \in \Re$ are $n$ positive real numbers between $0$ and $1$ and $D$ is a corresponding diagonal matrix of order $n \times n$, then the \textit{weighted inner product} of vectors $\vec{a}$ and $\vec{b}$ is defined by equation \ref{eqn3} \cite{Olver08}. It can be easily verified that \textit{weighted inner product} also satisfies the four basic properties of the \textit{inner product} mentioned earlier in this section. 
\begin{equation}\label{eqn3}
  \langle \vec{a}\,, \vec{b}\rangle_{D} = a^TDb = \sum_{i=1}^n d_ia_ib_i\,,  \text{where D=}
  \begin{bmatrix}
    d_1 & 0 & \dots & 0 \\
    0 & d_2 & \dots & 0 \\
    \vdots & \vdots & \ddots & \vdots \\
    0 & 0 & \dots & d_n
  \end{bmatrix}
\end{equation}
\subsubsection*{Norm:}
In linear algebra and related area of mathematics, the norm on a vector space $\Re^n$ is a function used to assign a non negative real number to each vector $\vec{a} \in \Re^n$. Every inner product gives a norm that can be used to calculate the length of a vector. However, not every norm is derived from an inner product \cite{Olver08}. The norm of a vector $\vec{a}\in \Re^n$ is denoted by $\lVert \vec{a} \rVert$ and can be defined using equation \ref{eqn4}.
\begin{equation}\label{eqn4}
  \lVert \vec{a} \rVert = \sqrt{\langle \vec{a}\,, \vec{a}\rangle}
\end{equation}
As given in \cite{Olver08}, if $\vec{a}=(a_1, a_2, ...,a_n)^T$ and $\vec{b}=(b_1, b_2, ...,b_n)^T$ are two vectors of the vector space $\Re^n$ and $n>0$ is a scalar quantity, then every norm satisfies the following three properties.
\begin{enumerate}
  \item  $\lVert \vec{a} \rVert > 0$ with $\lVert \vec{a} \rVert = 0$ iff $\vec{a}=0$ \hspace{40pt} \textit{(positivity)}
  \item $\lVert n\vec{a} \rVert =  n\lVert \vec{a} \rVert$ \hspace{91pt} \textit{(homogeneity)}
  \item $\lVert \vec{a} + \vec{b} \rVert \leq  \lVert \vec{a} \rVert + \lVert \vec{b} \rVert$ \hspace{33pt} \textit{(triangle inequality)}
\end{enumerate}
Some of the well-known norms defined in \cite{Weisstein02} are as follows:
\begin{enumerate}
  \item  \textit{L$_1$-norm (\textit{aka} Manhattan norm):} The \textit{$L_1$-norm} of a vector $\vec{a}\in \Re^n$ is simply obtained by adding the absolute values of its components. Formally, the \textit{L$_1$-norm} of the vector $\vec{a}$ is represented as $\lVert \vec{a} \rVert_1$ and it can be defined using equation \ref{eqn5}.
      \begin{equation}\label{eqn5}
         \lVert \vec{a} \rVert_1 = \sum_{i=1}^n |a_i|
      \end{equation}
  \item \textit{L$_2$-norm (\textit{aka} Euclidean norm):} The \textit{$L_2$-norm} of a vector $\vec{a}\in \Re^n$ is obtained by taking the positive square root of the sum of the square of its components. Formally, the \textit{L$_2$-norm} of a vector $\vec{a}$ is represented as $\lVert \vec{a} \rVert_2$ and it can be defined using equation \ref{eqn6}.
      \begin{equation}\label{eqn6}
         \lVert \vec{a} \rVert_2 = \left(\sum_{i=1}^n a_i^2\right)^{1/2}
      \end{equation}
   \item \textit{Infinity-norm (aka max-norm):} The \textit{Infinity-norm} of a vector $\vec{a}\in \Re^n$ is obtained as maximum of the absolute values of the components of the vector. Formally, the \textit{infinity-norm} of the vector $\vec{a}$ is represented as $\lVert \vec{a} \rVert_{\infty}$ and it can be defined using equation \ref{eqn7}.
      \begin{equation}\label{eqn7}
         \lVert \vec{a} \rVert_{\infty} = \max_{i=1}^n |a_i|
      \end{equation}
    \item \textit{L$_p$-norm:} The \textit{$L_p$ norm} of a vector $\vec{a}\in \Re^n$ is the positive $p^{th}$ root of the sum of $p^{th}$ power of the absolute value of the components of the vector. Formally, the \textit{$L_p$-norm} of the vector $\vec{a}$ is represented as $\lVert \vec{a} \rVert_p$ and it can be defined using equation \ref{eqn8}.
      \begin{equation}\label{eqn8}
         \lVert \vec{a} \rVert_p = \left(\sum_{i=1}^n |a_i|^p\right)^{1/p}
      \end{equation}
    \item \textit{Weighted Euclidean norm:} The \textit{weighted Euclidean norm} of a vector $\vec{a}\in \Re^n$ is a special case of the \text{Euclidean norm} in which different dimensions of $\vec{a}$ can have different weights. If $A^T=\vec{a}^T$ is a row matrix of order $1 \times n$, and $D$ is a diagonal matrix of order $n \times n$ whose $i^{th}$ diagonal element is the weight of the $i^{th}$ dimension of the vector $\vec{a}$, then the weighted Euclidean norm of  $\vec{a}$ is the positive square root of the product matrix $A^TDA$. In case $D$ is an identity matrix, this gives simply the \textit{Euclidean norm}. Formally, for a given vector $\vec{a}\in \Re^n$ and a weight (diagonal) matrix $D$ of order $n \times n$,  the \textit{weighted Euclidean norm} of $\vec{a}$ is represented as $\lVert \vec{a} \rVert_D$ and defined using equation \ref{eqn9}.
    
      \begin{equation}\label{eqn9}
         \lVert \vec{a} \rVert_D = \left(A^TDA\right)^{1/2} = \left(\sum_{i=1}^n d_ia_i^2\right)^{1/2}
      \end{equation}
   \end{enumerate}
\subsection{Multi-Attributed Graph}
In this section, we present a formal definition of multi-attributed graph, in which both vertices and edges can be represented as multi-dimensional vectors. Mathematically, a multi-attributed graph can be defined as a quadruple $G_m=(V, E, \mathcal{L}_v, \mathcal{L}_e)$, where $V \neq \phi$ represents the set of vertices, $E \subseteq V \times V$ represents the set of edges, $\mathcal{L}_v : V \to \Re^n$ is a vertex-label function that maps each vertex to an $n$-dimensional real vector, and $\mathcal{L}_e : E \to \Re^m$ is an edge-label function that maps each edge to an $m$-dimensional real vector. Accordingly, a vertex $v \in V$ in a multi-attributed graph can be represented as an $n$-dimensional vector $\vec{v} = (v_1, v_2, \dots, v_n)^T$ and an edge between a vertex-pair $(u, v)$ can be represented as an $m$-dimensional vector $\vec{e}(u, v) = (e_1(u, v), e_2(u, v), \dots, $ $e_m(u, v))^T$.
\subsection{Proposed Weighted Distance Measure}
In this section, we present our proposed weighted distance measure for multi-attributed graph that is based on \textit{weighted Euclidean norm}. If $u=\vec{u} =(u_1, u_2, \dots, u_n)^T$ and $v=\vec{v} =(v_1, v_2, \dots, v_n)^T$ are two vertices of a multi-attributed graph and $(u, v) = \vec{e}(u, v) = (e_1(u, v), e_2(u, v), \dots, $ $e_m(u, v))^T$ is a multi-labelled edge between $u$ and $v$, then the distance between $u$ and $v$ is denoted by $\Delta(u, v)$ and calculated using equation \ref{eqn10}, where $\lambda$ is a scalar value (equation \ref{eqn11}), and $I$ is an identity matrix of order $n \times n$. The value of $\lambda$  depends on the aggregate weight, $\omega(u, v)$, of the edges between the vertex-pair $(u, v)$, which is calculated using equation \ref{eqn12}. The value of $\gamma >0$ is a user-supplied threshold, providing flexibility to tune the value of $\lambda$ for calculating distance between a vertex-pair. In equation \ref{eqn12}, $\alpha_i$ is a constant such that $\alpha_i \geq 0$ and $\sum_{i=1}^n \alpha_i = 1$.

\begin{equation}\label{eqn10}
 \begin{split}
         \Delta(u, v)=\Delta(\vec{u}, \vec{v}) = \lVert \vec{u} - \vec{v} \rVert_{\lambda I} = \left((u - v)^T. \lambda I. (u - v)\right)^{1/2} \\
          = \left(
         \begin{bmatrix}
            u_1-v_1 & \dots & u_n-v_n
          \end{bmatrix}
          \begin{bmatrix}
            \lambda & \dots & 0 \\
            \vdots & \ddots & \vdots \\
            0 & \dots & \lambda
         \end{bmatrix}
         \begin{bmatrix}
            u_1-v_1 \\
            \vdots  \\
            u_n-v_n
         \end{bmatrix}
         \right)^{1/2}\\
          = \left(\lambda(u_1 - v_1)^2 + \lambda(u_2 - v_2)^2 + \dots + \lambda(u_n - v_n)^2\right)^{1/2} \\
          = \sqrt{\lambda} \times \left(\sum_{i=1}^n (u_i - v_i)^2\right)^{1/2}
         \end{split}
      \end{equation}
  \begin{equation}\label{eqn11}
     \lambda = \frac{1}{(1+\omega(u, v))^\gamma}
  \end{equation}
  \begin{equation}\label{eqn12}
  \begin{split}
     \omega(u, v) = \alpha_1 e_1(u, v) + \alpha_2 e_2(u, v) + \dots + \alpha_m e_m(u, v)\\ = \sum_{i=1}^n \alpha_i e_i(u, v)
  \end{split}
  \end{equation}

It may be noted that the $\lambda$ function defined using equation \ref{eqn11} is a monotonic decreasing function, as proved in theorem \ref{th1}, so that the distance between a vertex-pair could decrease with increasing ties between them, and vice-versa. The novelty of the proposed distance function lies in providing flexibility (i) to assign different weights to individual edges using $\alpha$, and (ii) to control the degree of monotonicity using $\gamma$, as shown in figure \ref{Lambda_Gamma}. It may also be noted that the value of $\lambda$ will always be 1, if either if $\omega(u, v)=0$ or $\gamma=0$.  


\begin{theorem}\label{th1}
The $\lambda$ function defined in equation \ref{eqn11} is a monotonically decreasing function. 
\end{theorem}

\begin{proof}
Let $\lambda = f(\omega)$ be a function of $\omega$, where $\omega$ is the aggregate weight of edges between a vertex-pair $(u, v)$. \\
Let $\omega_1, \omega_2 \geq 0$ such that $\omega_1 \leq \omega_2 $ \\
$\Rightarrow (1 + \omega_1) \leq (1 + \omega_2)$ \\
$\Rightarrow (1 + \omega_1)^{\gamma} \leq (1 + \omega_2)^{\gamma}$, where $\gamma \geq 1$ \\
$\Rightarrow \frac{1}{(1 + \omega_1)^{\gamma}} \geq \frac{1}{(1 + \omega_2)^{\gamma}}$\\ 
$\Rightarrow f(\omega_1) \geq f(\omega_2)$\\ 
Hence, $\lambda$ is a monotonically decreasing function of $\omega$ for $\gamma \geq 1$. \\
Similarly, it can be shown that $\lambda$ is a monotonically decreasing function of $\gamma$ for $\omega > 0$.\\
Let $\lambda = f(\gamma)$ be a function of $\gamma$\\
Let $\gamma_1, \gamma_2 \geq 1$ such that $\gamma_1 \leq \gamma_2$ \\
$\Rightarrow (1 + \omega)^{\gamma_1} \leq (1 + \omega)^{\gamma_2}$, where $ \omega > 0$ \\
$\Rightarrow \frac{1}{(1 + \omega)^{\gamma_1}} \geq \frac{1}{(1 + \omega)^{\gamma_2}}$\\ 
$\Rightarrow f(\gamma_1) \geq f(\gamma_2)$\\
Hence, $\lambda$ is a monotonically decreasing function of $\gamma$ for $\omega>0$.
\end{proof}

\begin{figure}[h]
\includegraphics[width=8cm]{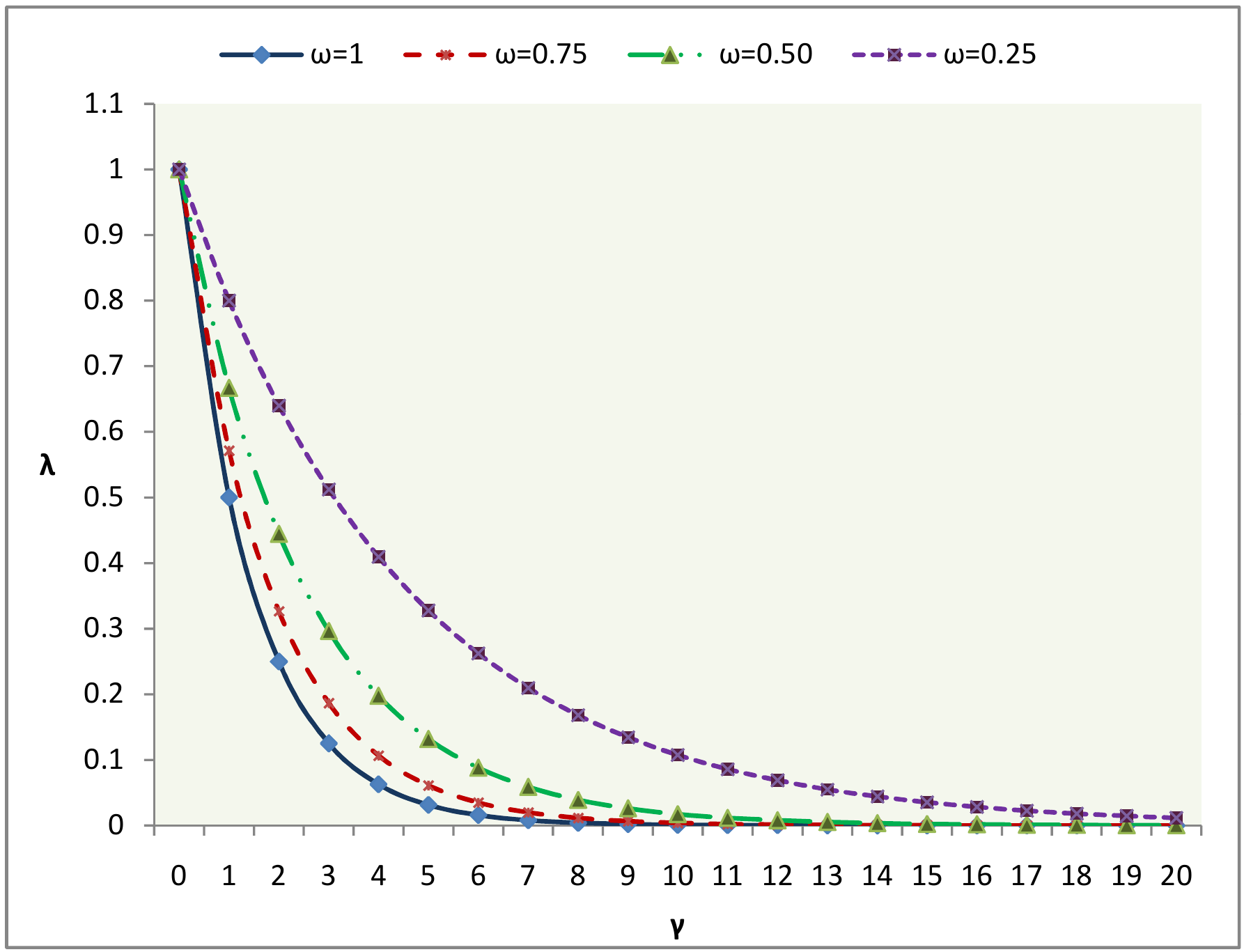}
\caption{Visualization of the  monotonically decreasing nature of $\lambda$ for varying $\gamma$ and $\omega$ values}
\label{Lambda_Gamma}
\end{figure}

It should be noted that if there is no direct link (edge) between a vertex pair $(u, v)$, then the value of $\omega(u, v)$ in equation \ref{eqn11} becomes zero, leading the value of $\lambda$ to 1. In this case, the value of $\Delta(u, v)$ is simply an \textit{Euclidean distance} between the vertex-pair $(u, v)$, as proved in theorem \ref{th2}.

\begin{theorem}\label{th2}
In a multi-attributed graph, if there is no edge between a vertex-pair then the distance between them is simply an Euclidean distance.
\end{theorem}
 
\begin{proof}
Let $u=\vec{u} =(u_1, u_2, \dots, u_n)^T$ and $v=\vec{v} =(v_1, v_2, \dots, v_n)^T$ be two vertices not connected by any edge. 
Since there is no edge between the vertex-pair $(u, v)$, the edge vector $\vec{e}(u, v) = (e_1(u, v), e_2(u, v), \dots, $ $e_m(u, v))^T = \vec{0}$.\\ 
$\Rightarrow e_1(u, v) = e_2(u, v) =  \dots = $ $e_m(u, v) = 0$\\
Hence, $\omega(u, v) = \alpha_1 . e_1(u, v) + \alpha_2 . e_2(u, v) + \dots + \alpha_m . e_m(u, v)\\ 
= \alpha_1 . 0 + \alpha_2 . 0 + \dots + \alpha_m . 0 = 0.$  \hspace{10 pt} [using equation 12]\\
Hence, $\lambda = \frac{1}{(1+\omega(u, v))^\gamma} = \frac{1}{(1+0)^\gamma} = 1$ \hspace{10 pt} [using equation 11] \\ 
Finally, $\Delta(u, v)= \sqrt{\lambda} \times \left(\sum_{i=1}^n (u_i - v_i)^2\right)^{1/2}$ \hspace{10 pt}  [using equation 10]\\
$= \sqrt{1}\times \left(\sum_{i=1}^n (u_i - v_i)^2\right)^{1/2} = \left(\sum_{i=1}^n (u_i - v_i)^2\right)^{1/2}$, which is an \textit{Euclidean distance} between the vertex-pair ($u$, $v$).
\end{proof}

Algorithm \ref{algorithm1} presents a formal way to calculate the distance between all pairs of vertices of a given multi-attributed graph using \emph{MAGDist}. The proposed \emph{MAGDist} algorithm reads a multi-attributed graph using two separate CSV files -- one containing the list of vertex vectors and the other containing the list of edge vectors, and produces the distance between each vertex-pairs as a CSV file, wherein each tuple contains a vertex-pair and distance value.  We have also proposed \emph{MAGSim} algorithm to generate similarity graph using the distance values calculated by the \emph{MAGDist} algorithm. The proposed algorithm reads each vertex-pair and its distance value $<i, j, \Delta(i, j)>$ and calculates similarity between the vertex-pair $(i, j)$ using equation \ref{eqn14}.  
\begin{equation}\label{eqn14}
     sim(i, j) = 1 -  \frac{\Delta(i, j)}{\max\limits_{x,y \in V}\{\Delta(x, y)\}}
\end{equation}
\begin{algorithm}[htbp]
    \tcp{computing distance between all pairs of vertices of a multi-attributed graph}
    \SetKwInOut{Input}{Input}
    \SetKwInOut{Output}{Output}
    \Input{CSV files $\mathcal{L}_v$ and $\mathcal{L}_e$ containing list of vertex vectors, and edge vectors, respectively. An $1D$ array $\alpha[1..m]$, wherein $\alpha_i \geq 0$ and $\sum_{i=1}^m \alpha_i = 1$ for calculating the linear combination of edge weights between a pair of vertices. A positive integer threshold $\gamma$ representing the weightage of edges in distance calculation.}
    \Output{A CSV file $D$ containing distance between each pairs of vertices.}
    $n_v \leftarrow vertexCount[\mathcal{L}_v]$ \tcp*{number of vertices}
    $n \leftarrow vertexDimCount[\mathcal{L}_v]$ \tcp*{vertex vector dimension}
    $m \leftarrow edgeDimCount[\mathcal{L}_e]$ \tcp*{edge vector dimension}
    $V[n_v][n] \leftarrow \mathcal{L}_v$ \tcp*{reading $\mathcal{L}_v$ into array $V$.}
    \For{each vertex-pair $(i, j) \in $ $\mathcal{L}_e$}{
       \tcp{calculating aggregate weight $\omega(i, j)$ of the edges between vertex-pair  $(i, j)$.}
       $\omega(i, j) \leftarrow 0$\;
       \For{$k\leftarrow 1$ \KwTo $m$}{
         $\omega(i, j) = \omega(i, j) + \alpha[k] \times e_k(i, j)$ \tcp*{[eqn. \ref{eqn12}]}
       }
\tcp{calculating the value of scalar quantity $\lambda$}      
       $\lambda = \frac{1}{(1 + \omega(i, j))^\gamma}$ \tcp*{[eqn. \ref{eqn11}]}
       \tcp{calculating distance $\Delta(i, j)$ between the vertex-pair $(i, j)$.} 
       $d \leftarrow 0$\;
       \For{$k\leftarrow 1$ \KwTo $n$}{
            $d = d + (V[i][k] - V[j][k])^2$\;
       }
       $\Delta(i, j) = \sqrt{\lambda} \times \sqrt{d}$ \tcp*{[eqn. \ref{eqn10}]}
       write tuple $<i, j, \Delta(i, j)>$ into $D$ \;
    }
    \caption{\textbf{MAGDist$(\mathcal{L}_v, \mathcal{L}_e, \alpha, \gamma)$}}
    \label{algorithm1}
\end{algorithm}

\begin{algorithm}[htbp]
    \tcp{generating similarity graph corresponding to a multi-attributed graph}
    \SetKwInOut{Input}{Input}
    \SetKwInOut{Output}{Output}
    \Input{A CSV file $D$ containing vertex-pairs along with distance output by the \emph{MAGDist}.}
    \Output{A CSV file $G_s$ containing edge-listed similarity graph.}
    $d_{max} \leftarrow$ getMaxDistance($D$)\; 
    \For{each tuple $<i, j, \Delta(i, j>$ in $D$}{
       $sim(i, j) = 1 -  \frac{\Delta(i, j)}{d_{max}}$
       write tuple $<i, j, sim(i, j)>$ into $G_s$ \;
    }
    \caption{\textbf{MAGSim$(D)$}}
    \label{algorithm2}
\end{algorithm}
\textbf{Example:} Figure \ref{fig1} presents a simpler case of multi-attributed graph having four vertices in which each vertex is represented as a two dimensional feature vector and each edge is represented as an one dimensional vector. In case, there is no direct edge between a pair of vertices (e.g., $v_1$ and $v_3$), the corresponding edge vector is a zero vector. If we simply calculate the Euclidean distance between the vertex-pairs, then the distance between the vertex-pairs $(v_1, v_2), (v_2, v_3), (v_3, v_4)$ and $(v_4, v_1)$ is same (10 unit), whereas the distance calculated by the proposed \emph{MAGDist} differs, based on the weight of the edges between them. Similarly, the Euclidean distance between the vertex-pairs $(v_1, v_3)$ and $(v_2, v_4)$ are same ($10\sqrt{2}$), but the distance values calculated using \emph{MAGDist} are different. The distance values between each vertex-pairs of the multi-attributed graph calculated using \emph{MAGDist} for $\gamma=1$ and $\gamma=2$ are is shown in $D_1$ and $D_2$ matrices, respectively of figure \ref{fig1}. It can be observed that giving more weightage to edges by increasing the value of $\gamma$ reduces the distance between the respective vertex-pairs.
\begin{figure*} [htbp]
\includegraphics[width=0.8\textwidth, clip=true]{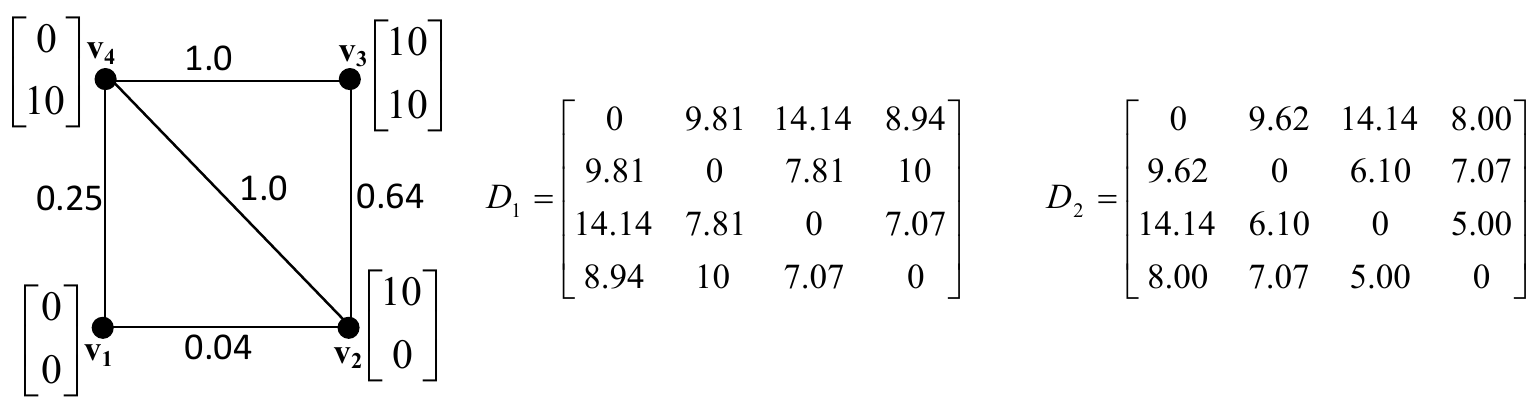}
\caption{A multi-attributed graph with vertices as multi-dimensional vectors, and distance matrices $D_1$, $D_2$ calculated using \emph{MAGDist} algorithm for $\gamma=1$ and $\gamma=2$, respectively}
\label{fig1}
\end{figure*}
\begin{figure*}[htbp]
\includegraphics[width=0.8\textwidth, clip=true]{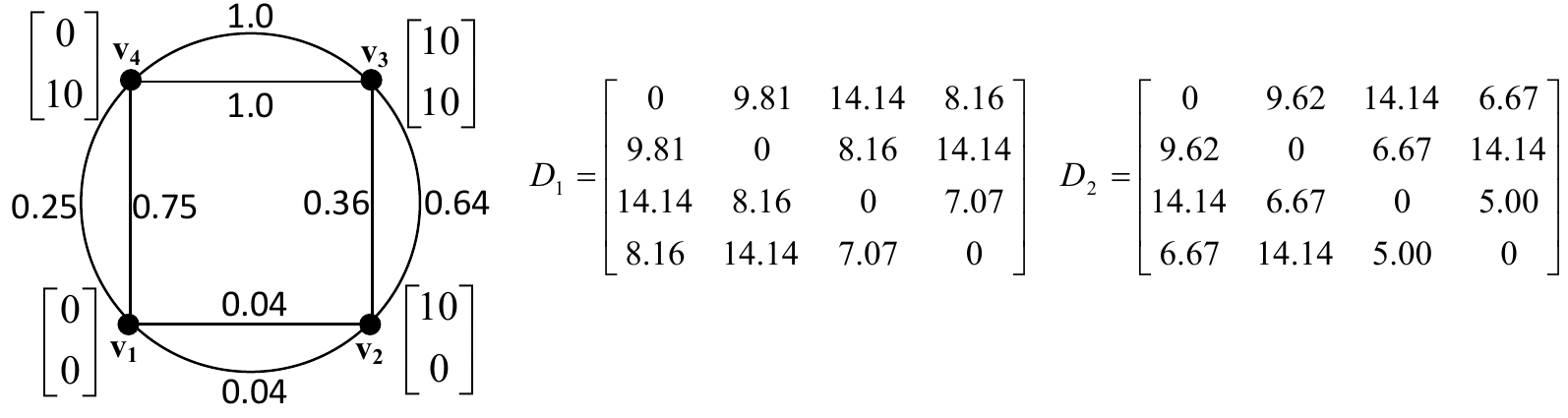}
\caption{A multi-attributed graph with both vertices and edges as multi-dimensional vectors, and distance matrices $D_1$, $D_2$ calculated using \emph{MAGDist} algorithm for $\gamma=1$ and $\gamma=2$, respectively}
\label{fig2}
\end{figure*}
Figure \ref{fig2} presents a multi-attributed graph in which both vertex and edge labels are multi-dimensional vectors. For example, the edge vector corresponding to the edge connecting $v_2$ and $v_3$ is $\vec{e}(2, 3) = (0.36, 0.64)^T$. If we simply calculate the Euclidean distance between the vertex-pairs, then the distance between the vertex-pairs $(v_1, v_2), (v_2, v_3), (v_3, v_4)$ and $(v_4, v_1)$ is same (10 unit), whereas the distance calculated by the proposed \emph{MAGDist} differs, based on the weight of the edges between them. 
The distance values between each vertex-pairs of the multi-attributed graph calculated using \emph{MAGDist} for $\gamma=1$ and $\gamma=2$ are is shown in $D_1$ and $D_2$ matrices, respectively of figure \ref{fig2}. It can be observed from these two distance matrices too that giving more weightage to edges by increasing the value of $\gamma$ reduces the distance between the respective vertex-pairs.
\section{Experimental Setup and Results}
To establish the efficacy of the proposed MAGDist distance measure, we have considered the well-known Iris data set \footnote{\url{http://archive.ics.uci.edu/ml/datasets/Iris}}, which contains total 150 instances of three different categories of Iris flower; 50 instances of each category. The Iris data set can be represented as a $150\times5$ data matrix, wherein each row represents an Iris flower, first four columns represent four different attribute values in centimetre, and the $5^{th}$ column represents class labels (categories) of the Iris flower as \textit{Setosa}, \textit{Virginica}, or \textit{Versicolour}.  Table \ref{Table1} shows a partial snapshot of the Iris data set.
\begin {table}
\caption{A partial view of the Iris data set}
 \label{Table1}
\begin{center}
\scalebox{0.83}{
 \begin{tabular}{|c|c|c|c|c|} 
 \hline
 \textbf{Sepal length} & \textbf{Sepal width} & \textbf{Petal length} & \textbf{Petal width} & \textbf{Species}\\ \hline
 5.1 & 3.5 & 1.4 & 0.2 & Setosa\\ \hline
4.9 & 3.0 & 1.4 & 0.2 & Setosa\\ \hline
4.7 & 3.2 & 1.3 & 0.2 & Setosa\\ \hline
4.6 & 3.1 & 1.5 & 0.2 & Setosa\\ \hline
5.0 & 3.6 & 1.4 & 0.2 & Setosa\\ \hline
... & ... & ... & ... & ...\\ \hline
6.4 & 3.2 & 4.5 & 1.5 &  versicolor\\ \hline
6.9 & 3.1 & 4.9 & 1.5 & versicolor\\ \hline
5.5 & 2.3 & 4.0 & 1.3 & versicolor\\ \hline
6.5 & 2.8 & 4.6 & 1.5 & versicolor\\ \hline
5.7 & 2.8 & 4.5 & 1.3 & versicolor\\ \hline
... & ... & ... & ... & ...\\ \hline
6.3 & 3.3 & 6.0 & 2.5 & virginica\\ \hline
5.8 & 2.7 & 5.1 & 1.9 & virginica\\ \hline
7.1 & 3.0 & 5.9 & 2.1 & virginica\\ \hline
6.3 & 2.9 & 5.6 & 1.8 & virginica\\ \hline
6.5 & 3.0 & 5.8 & 2.2 & virginica\\ \hline
\end{tabular}}
\end{center}
\end{table}
We model Iris data set as a multi-attributed similarity graph in which each vertex $v \in \Re^4$ is a 4-dimensional vector representing a particular instance of the Iris flower. The edge (similarity) between a vertex-pair $(u, v)$ is determined using equation \ref{eqn31} on the basis of the Gaussian kernel value defined in equation \ref{eqn13}, where $\sigma=1$ is a constant value \cite{Zaki2014:DMA}. 
\begin{equation}\label{eqn31}
\mathit{e(u, v)=
\begin{cases}
\kappa^G(u, v) & \text{if }\kappa^G(u, v) \geq 0.55\\
0 & \text{otherwise}
\end{cases}
}
\end{equation}
\begin{equation}\label{eqn13}
     \kappa^G(u, v) = e^{- \frac{ \lVert u - v \rVert^2}{2\sigma^2}}
\end{equation}
The resulting multi-attributed Iris similarity graph is shown in figure \ref{fig3} (termed hereafter as $G_1$ for rest of the paper), which contains 150 vertices and 2957 edges. In this graph, the instances of \textit{Setosa}, \textit{Virginica}, or \textit{Versicolour} are shown using \textit{triangles}, \textit{squares}, and \textit{circles}, respectively. Although $\kappa^G(u, u)=1.0$ for all vertices, we haven't shown self-loops in this graph. The proposed \emph{MAGDist} algorithm is applied over $G_1$ to calculate distance between all vertex-pairs, and finally \emph{MAGSim} algorithm is applied to generate Iris similarity graph (termed hereafter as $G_2$ for rest of the paper), which is shown in figure \ref{fig44}. In \cite{Zaki2014:DMA}, the authors have also used Gaussian kernel followed by the concept of nearest-neighbours to generate Iris similarity graph (termed hereafter as $G_3$ for rest of the paper) and applied Markov Clustering (MCL) to classify the instances of the Iris data into three different categories. Therefore, in order to verify the significance of our \emph{MAGDist} and \emph{MAGSim} algorithms, we have also applied the MCL over the Iris similarity graphs $G_1$ and $G_2$, and present a comparative analysis of all clustering results. Figures \ref{fig4} and \ref{fig5} present the clustering results after applying MCL over the Iris similarity graphs $G_1$ and $G_2$, respectively. Table \ref{Table2} presents the contingency table for the discovered clusters versus true Iris types from all three different Iris similarity graphs. It can be observed from this table that, in case of $G_2$, only six instances of iris-versicolor are wrongly grouped with iris-virginica in $C_2$ and three instances of iris-virginica are wrongly grouped with iris-versicolor in $C_3$; whereas in $G_1$ forty iris-versicolor instances are wrongly grouped with iris-virginica in $C_2$, and in $G_3$, one instance of iris-versicolor is wrongly grouped with iris-setosa in $C_1$ and fourteen instances of iris-virginica are wrongly grouped with iris-versicolor in $C_3$. 
\begin{figure}
\includegraphics[width=6cm]{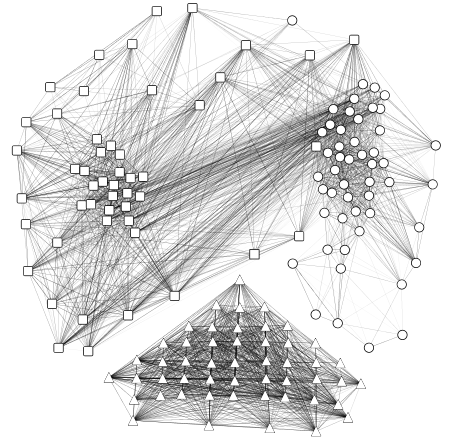}
\caption{\textbf{$G_1$:} Iris data graph modelled as a multi-attributed graph in which vertices are 4-dimensional real vectors and edges are the Gaussian similarity ($\geq 0.55$) between the respective vertices}
\label{fig3}
\end{figure}
\begin{figure}[h]
\includegraphics[width=6cm]{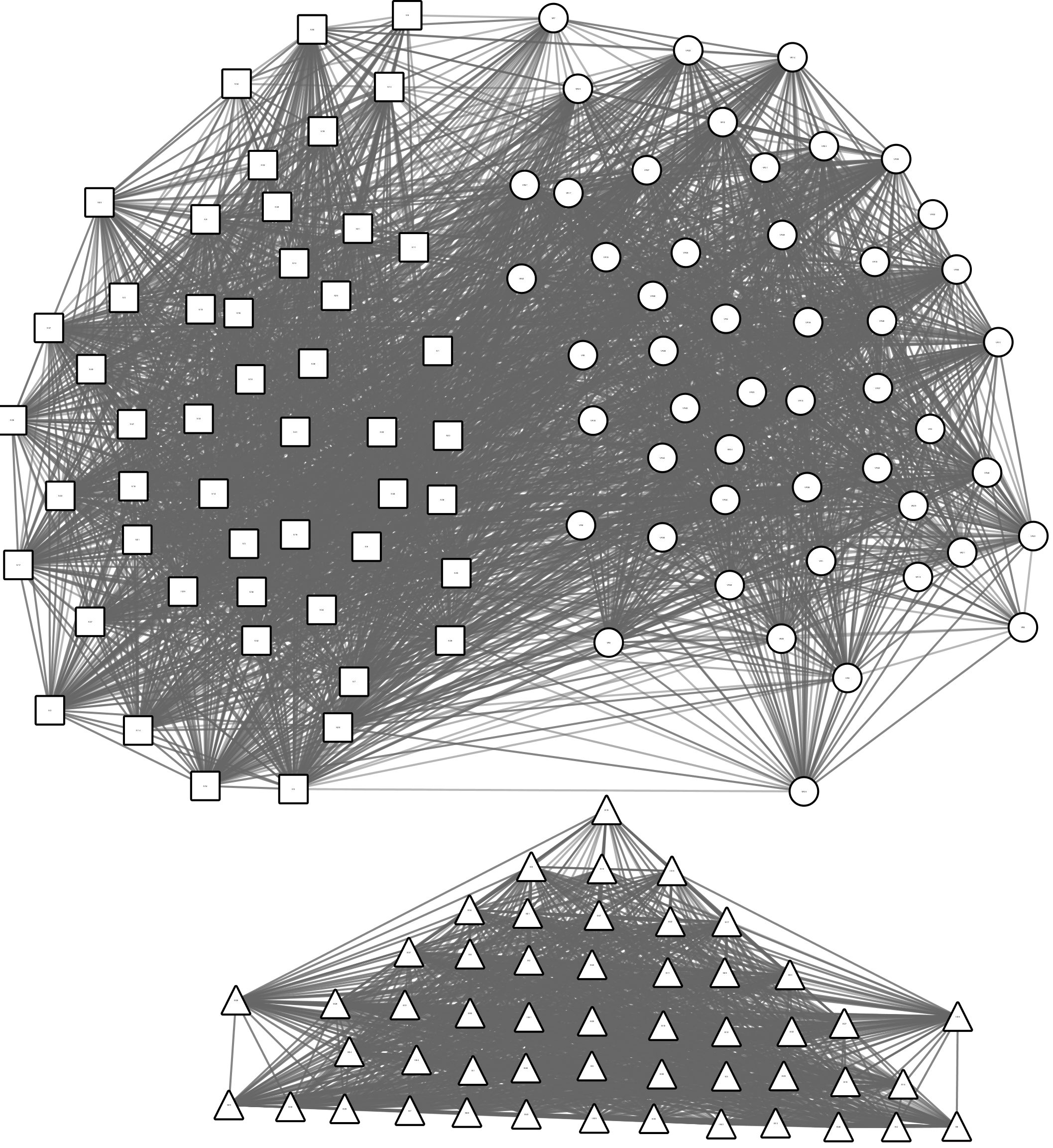}
\caption{\textbf{$G_2$:} Iris data graph modelled as a multi-attributed graph in which vertices are 4-dimensional real vectors and edges are \emph{MAGSim} similarity ($\geq 0.80$) calculated using equation \ref{eqn14}}
\label{fig44}
\end{figure}
\begin{figure}[h]
\includegraphics[width=6cm]{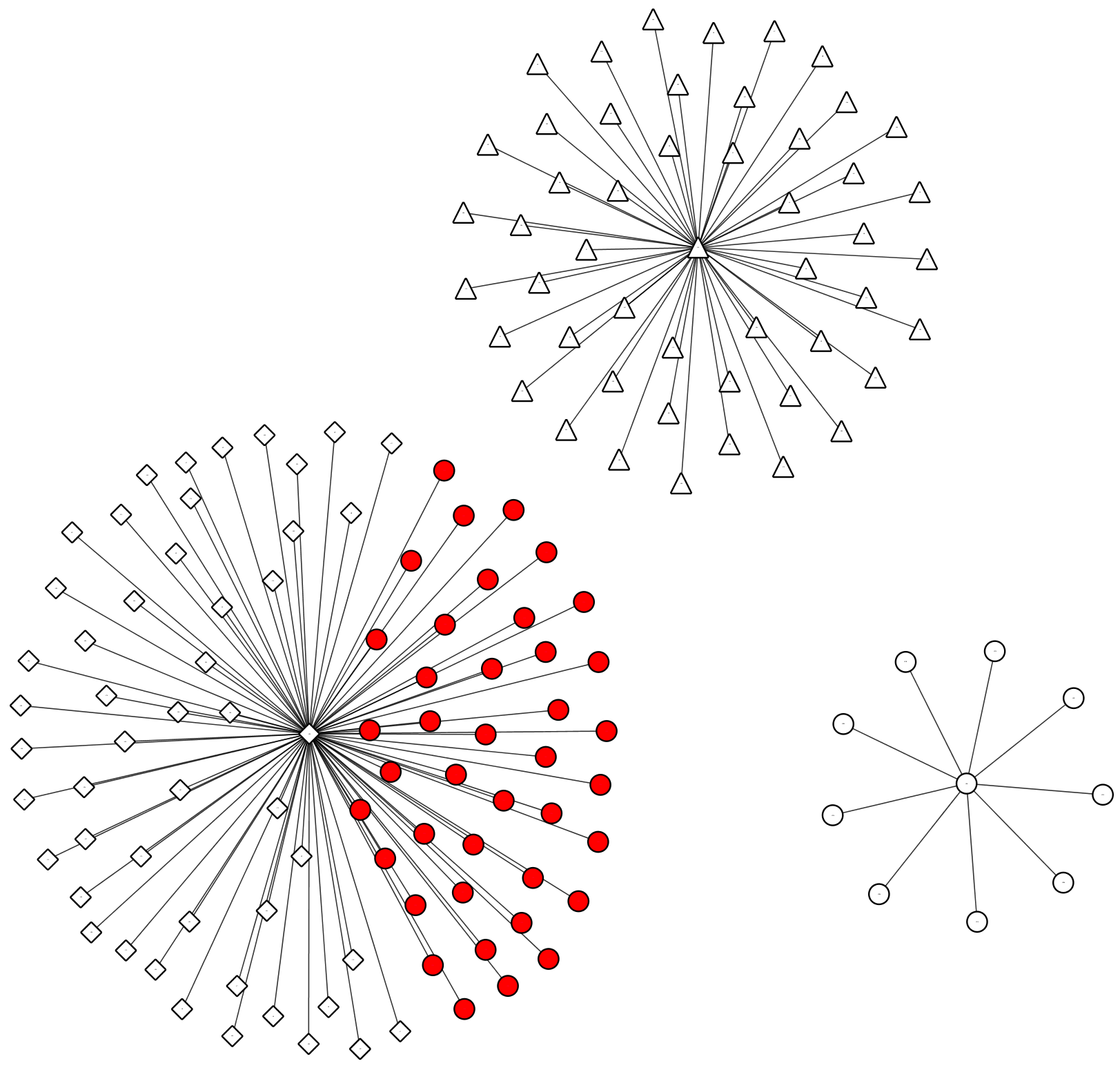}
\caption{Clustering results after applying MCL over the Iris data graph $G_1$}
\label{fig4}
\end{figure}
\begin{figure}[h]
\includegraphics[width=8cm]{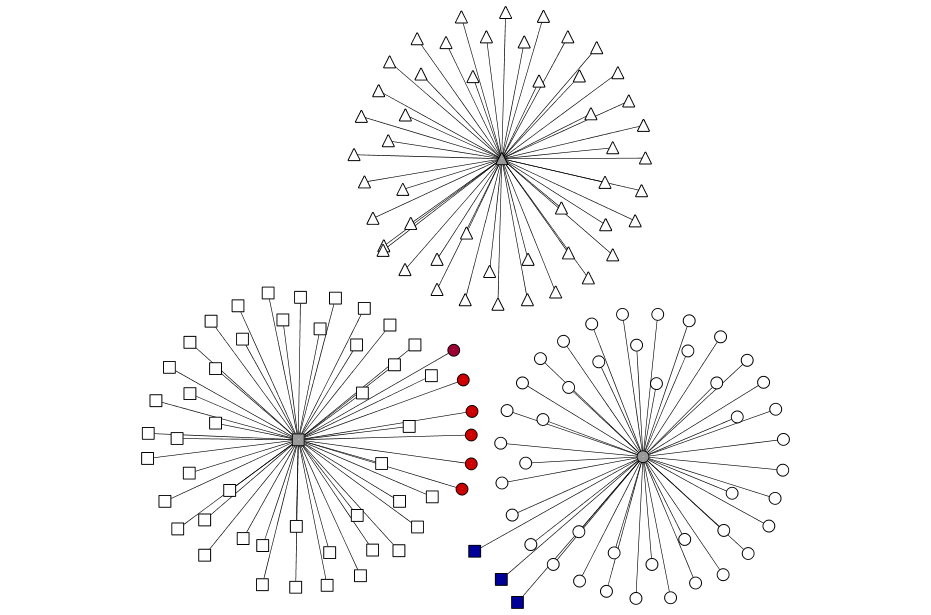}
\caption{Clustering results after applying MCL over the Iris data graph $G_2$}
\label{fig5}
\end{figure}

We have also analyzed the significance of different similarity graph generation methods in terms of \textit{True Positive Rate} (\emph{TPR}) and \textit{False Positive Rates} (\emph{FPR}) that are defined using equations \ref{eqn15} and \ref{eqn16}, respectively. In these equations, \emph{TP} is the \textit{True Positives}, representing the number of correctly classified positive instances, \emph{FP} is the \textit{False Positives}, representing the number of wrongly classified negative instances, and \emph{P} and \emph{N} represent the total number of positive and negative instances, respectively in the data set. Table \ref{Table3} presents \emph{TPR} and \emph{FPR} values for all three different similarity graphs, showing best results for $G_2$, which has been generated using our proposed \emph{MAGDist} and \emph{MAGSim} algorithms.
 
\begin {table}
\caption{Contingency table for MCL clusters from similarity graphs generated by three different methods}
\label{Table2}
\begin{center}
\scalebox{0.9}{
\begin{tabular}{|l|c c c|c c c|c c c|} \hline
 \multirow{2}{*}{\textbf{Clusters }} & \multicolumn{3}{c|}{\textbf{$G_1$ }} & \multicolumn{3}{c|}{\textbf{$G_2$}} & \multicolumn{3}{c|} {\textbf{$G_3$} \cite{Zaki2014:DMA}} \\ \cline{2-10}
& \textbf{Set} & \textbf{Vir} & \textbf{Ver} & \textbf{Set} & \textbf{Vir} & \textbf{Ver} & \textbf{Set} & \textbf{Vir} & \textbf{Ver} \\ \hline
$C_1$ (triangle) & 50 & 0 & 0 & 50 & 0 & 0 & 50 & 0 & 1 \\ \hline
$C_2$ (square) & 0 & 50 & 40 & 0 & 47 & 6 & 0 & 36 & 0 \\ \hline
$C_3$ (circle) & 0 & 0 & 10 & 0 & 3 & 44 & 0 & 14 & 49 \\ \hline
\end{tabular}}
\end{center}
\end{table}

\begin {table}[htbp]
\caption{Performance comparison of three different methods on Iris data set}
\label{Table3}
\begin{center}
\scalebox{0.9}{
\begin{tabular}{|l|c c|c c|c c|} \hline
 \multirow{2}{*}{\textbf{Clusters }} & \multicolumn{2}{c|}{\textbf{$G_1$}} & \multicolumn{2}{c|}{\textbf{$G_2$}} & \multicolumn{2}{c|} {\textbf{$G_3$}\cite{Zaki2014:DMA}} \\ \cline{2-7}
& \textbf{TPR} &  \textbf{FPR} & \textbf{TPR} & \textbf{FPR} & \textbf{TPR} & \textbf{FPR} \\ \hline
$C_1$ (triangle) & 1.00 & 0.00 & 1.00 & 0.00 & 1.00 & 0.01 \\ \hline
$C_2$ (square) & 1.00 & 0.40 & 0.94 & 0.06 & 0.72 & 0.00  \\ \hline
$C_3$ (circle) & 0.20 & 0.00 & 0.88 & 0.03 & 0.98 & 0.14 \\ \hline
\textbf{Average} & \textbf{0.73} &\textbf{ 0.13} & \textbf{0.94} & \textbf{0.03} & \textbf{0.90} & \textbf{0.05} \\ \hline
\end{tabular}}
\end{center}
\end{table}

\begin{equation}\label{eqn15}
     \text{TPR} = \frac{TP}{P}
\end{equation}
\begin{equation}\label{eqn16}
     \text{FPR} =  \frac{FP}{N}
\end{equation}

\subsection{Evaluation Results on Twitter Data Set}
In order to illustrate the application of the proposed distance measure over a real multi-attributed graph, we have considered a Twitter data set of 300 tweets, comprising equal number of tweets related to three different events -- \textit{NoteBandi (NTB)}, \textit{RyanInternationalSchool (RIS)}, and \textit{Rohingya (ROH)}. The tweets are modelled as a multi-attributed graph, wherein each vertex represents a tweet as an 110-dimensional binary vector based on the top-110 key-terms identified using \textit{tf-idf}, and two different edges exist between a vertex-pair -- one representing the degree of \textit{Hashtags} overlap and the other representing the \textit{tweet-time} overlap. The similarity graph is generated using \emph{MAGSim} algorithm, and shown in figure \ref{figTweetSim} wherein the instances of \textit{NoteBandi}, \textit{RyanInternationalSchool}, and \textit{Rohingya} are represented using \textit{squares}, \textit{triangles}, and \textit{circles}, respectively. Finally, \emph{MCL} is applied over the similarity graph to group the tweets into different clusters shown in figure \ref{figTweetMCL}. The evaluation of the obtained clusters is given in table \ref{Table4}. It can be seen from this table that only five instances of \textit{RyanInternationalSchool} are wrongly clustered with \textit{Rohingya} in $C_3$.

\begin{figure}[h]
\includegraphics[width=6cm]{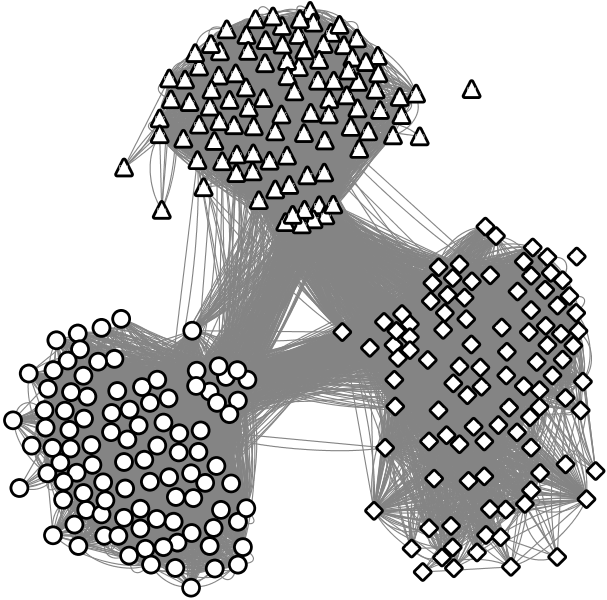}
\caption{Similarity graph generated from Twitter data set (\textit{only edges having similarity value > 0.5 are shown})}
\label{figTweetSim}
\end{figure}
\begin{figure}[h]
\includegraphics[width=8cm]{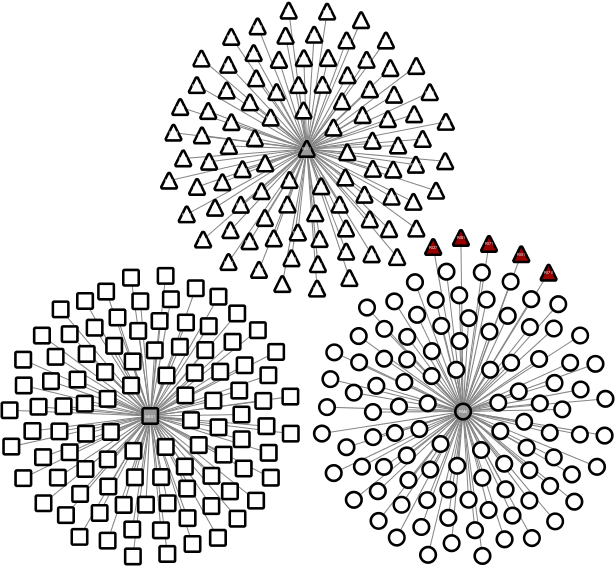}
\caption{Clustering results obtained by applying MCL over the similarity graph of the Twitter data set}
\label{figTweetMCL}
\end{figure}

\begin {table}
\caption{Evaluation results of the proposed method on Twitter data set}
\label{Table4}
\begin{center}
\scalebox{0.9}{
\begin{tabular}{|l|c c c|c c|} \hline
 \multirow{2}{*}{\textbf{Clusters}} & \multicolumn{3}{c|}{\textbf{Contingency Table}} & \multicolumn{2}{c|}{\textbf{Evaluation Results}}  \\ \cline{2-6}
& \textbf{NTB} &  \textbf{RIS} & \textbf{ROH} & \textbf{TPR} & \textbf{FPR} \\ \hline
$C_1$ (triangle) & 0 & 95 & 0 & 1.00 & 0.00  \\ \hline
$C_2$ (square) & 100 & 0 & 0 & 1.00 & 0.00  \\ \hline
$C_3$ (circle) & 0 & 5 & 100 & 1.00 & 0.025 \\ \hline
\multicolumn{4}{|r}\textbf{Average} & \textbf{1.00} & \textbf{0.008}  \\ \hline
\end{tabular}}
\end{center}
\end{table}

\section{Conclusion and Future Works}\label{conclusion}
In this paper, we have proposed a novel weighted distance measure based on weighted Euclidean norm that can be used to calculate the distance between vertex-pairs of a multi-attributed graph containing multi-labelled vertices and multiple edges between a single vertex-pair. The proposed distance measure considers both vertex and edge weight-vectors, and it is flexible enough to assign different weightage to different edges and scale the overall edge-weight while computing the weighted distance between a vertex-pair. We have also proposed a \emph{MAGDist} algorithm that reads the lists of vertex and edge vectors as two separate CSV files and calculates the distance between each vertex-pairs using the proposed weighted distance measure. Finally, we have proposed a multi-attributed similarity graph generation algorithm, \emph{MAGSim}, which reads the output produced by the \emph{MAGDist} algorithm and generates a similarity graph, which can be used by the existing classification and clustering algorithms for various analysis tasks. Since the proposed \emph{MAGDist} algorithm reads multi-attributed graph as CSV files containing vertex and edge vectors, it can be scaled to handle large complex graphs (aka big graphs). Applying proposed distance measure and algorithms on (research articles) citation networks and online social networks to analyze them at different levels of granularity is one of the future directions of work. 

%
%
\bibliographystyle{ACM-Reference-Format}
\bibliography{p42-abulaish}
\end{document}